\documentclass[copyright,creativecommons]{eptcs}
\providecommand{\event}{CCA 2010} % Name of the event you are submitting to
\providecommand{\volume}{??}
\providecommand{\anno}{2010}
\providecommand{\firstpage}{1}
\providecommand{\eid}{??}
\usepackage{breakurl}        % Not needed if you use pdflatex only.
\usepackage{amssymb}
\usepackage{graphicx}

\newtheorem{definition}{Definition}
\newtheorem{theorem}{Theorem}
\newtheorem{lemma}{Lemma}
\newtheorem{corollary}{Corollary}
\newtheorem{proposition}{Proposition}

\newenvironment{proof}{\parindent=0pt{\bf Proof: }}{
   \hspace*{\fill}\hbox to 6pt{\leaders\hrule width 6pt height 6pt\hfill}\par}

\long\def\remove#1{}

\newcommand{\pbox}{\hbox to 6pt{\leaders\hrule width 6pt height 6pt\hfill}}

\usepackage{color}

\title{NP-Logic Systems and Model-Equivalence Reductions
\thanks{Research was partially supported by the NSFC projects under
grant No. 60970040}}

\author{Yuping Shen
\institute{
Institute of Logic and Cognition\\ Sun Yat-Sen University
\\ 510275 Guangzhou, P. R. China}
 \email{shyping@mail.sysu.edu.cn}
 \and
Xishun Zhao
\institute{
Institute of Logic and Cognition\\ Sun Yat-Sen University
\\ 510275 Guangzhou, P. R. China}
\email{hsszxs@mail.sysu.edu.cn}\\
}
\def\titlerunning{NP-Logic Systems and Model-Equivalence Reductions}
\def\authorrunning{Y. Shen \& X. Zhao}
\begin{document}
\maketitle
%%%%%%%%%%%%%%%%%%%%%%%%%%%%%%%%%%%%%%%%%%%%%%%%%%%%%%%%%%%%

\maketitle

\begin{abstract}
In this paper we investigate the existence of model-equivalence
reduction between NP-logic systems which are logic systems with
\emph{model existence} problem in NP. It is shown that among all
NP-systems with \emph{model checking} problem in NP, the
existentially quantified propositional logic ($\exists$PF) is
maximal with respect to poly-time model-equivalent reduction.
However, $\exists$PF seems not a maximal NP-system in general
because there exits a NP-system with model checking problem
$D^P$-complete.
\end{abstract}

%{\bf Keywords:} Logic Sytem; Model-equivalent Reduction; Expressive Power; Complexity.

%%%%%%%%%%%%%%%%%%%
\section{Introduction}

For a complexity class ${\cal C}$, there are many logic systems for
which the \emph{model existence} problem  (i.e. the satisfiability
problem) lies in ${\cal C}$. We call such systems ${\cal
C}$-systems. Take NP as an example, the following logic systems are
all NP-systems:

\begin{itemize}
\item PF, the class of propositional formulas,
\item CNF, the class of propositional formulas in conjunctive normal form,
\item $k$CNF, the class  of CNF-formulas in which each clauses contains at
most $k$ literals, where $k\geq 3$.
\item LP, the class of normal logic programs with answer set
semantics \cite{gelfond}.

%\item PF(3), the class of propositional formulas with a 3-valued
%semantics (say e.g. Kleene semantics).

\item $\exists$PF, the class of quantified Boolean formulas with
only existential quantifiers.
\end{itemize}

Among the above systems we have the following observations:

\begin{itemize}
\item All systems in $\{$PF, CNF, LP, $\exists$PF$\}$ have the
same expressive power w.r.t. equivalence. More precisely, for any
two systems ${\cal S}_1, {\cal S}_2\in$ $\{$PF, CNF, LP,
$\exists$PF$\}$, there is a transformation which translates every
formula in ${\cal S}_1$ to a formula in ${\cal S}_2$ such that the
two formulas are equivalent (i.e., they have the same models).

\item From $(k+1)$CNF to $k$CNF there is no transformation which
preserves the equivalence. Further, CNF has strictly stronger
expressive power than $k$CNF (see e.g. \cite{kble}).

%\item PF(3) has strictly stronger expressive power than any system
%in $\{$ PF, CNF, $k$CNF, LP, $\exists$PF $\}$.

\item From PF to CNF there is no poly-space transformation which
preserves the equivalence (see e.g. \cite{kble}).

\item Under the conjecture $\mbox{P}\not\subseteq
{\mbox{NC}^1}/{\mbox{poly}}$ (see \cite{papa}), there is no
poly-space transformation from LP to PF which preserves the
equivalence \cite{lifraz}.
\end{itemize}

From the above we can see that the expressive power of logic systems
in the same complexity class are quite different. Since poly-space
transformations preserving equivalence do not exist between some
NP-systems, it is quite natural to investigate the existence of
(poly-time or poly-space) reductions which only preserve some
relaxed equivalence.  One of such reductions called model-equivalent
reduction was introduced by Xishun Zhao and Kleine B\"{u}ning in
\cite{zhkb}. Informally speaking, a system ${\cal S}$ can be
model-equivalently reduced to ${\cal S}'$ if every formula $F$ in
${\cal S}$ can be transformed into a formula $F'$ in ${\cal S}'$
such that there is a poly-time computable one-to-one correspondence
between the models of $F$ and $F'$. With respect to poly-time
model-equivalent reduction, systems PF, CNF, 3CNF, LP have the same
expressive power (see \cite{zhkb} \cite{linzhao}). However,
$\exists$PF still has strictly stronger expressive power than PF
under the conjecture that NP$\not\subseteq$P/poly which is widely
believed true. That is, there seems no even poly-space
model-equivalent reduction from $\exists$PF to PF. So, the authors
of \cite{zhkb} asked whether $\exists$PF is a maximal NP-system
w.r.t. poly-time model-equivalent reduction. This paper is concerned
with this question. The remainder is organized as follows. In
section 2, for the convenience of proof of our main result, we give
a general but formal definition of logic systems. After listing some
examples of logic systems, we reformulate the definition of
model-equivalent reduction. In section 3, the main results are
proved. More precisely, we prove the following: Any NP-system with
model checking problem in NP can be poly-time
model-equivalently to $\exists$PF, %and (2) any NP-system with model
%checking problem in P can be poly-time model-equivalently to PF.
The \emph{model checking} problem is to decide whether a given formula
is satisfied by a given interpretation. However, there do exist a
NP-system for which the model checking problem is co-NP-complete or
even harder. Then we show that there is a NP-system which is
incomparable with $\exists$PF, and that there is a NP-system with
strictly stronger expressive power than $\exists$PF under a
conjecture in complexity theory.

\section{Logic System and Model-equivalent Reduction}

Roughly speaking, a logic system consists of three parts, the
language which is usually identified with the class of formulas
constructed from the symbols in the language, the semantics which
consists of all possible interpretations of symbols in the language,
and the deductive relation. However, for our purpose we adopt the
following formal definition.

\begin{definition} \label{d-1} A logic system is a tuple $(\Gamma,
\Delta, T, S, R)$ satisfying the following conditions:
\begin{itemize}
\item $\Gamma, \Delta$ are non-empty finite sets of symbols, and $\Gamma\cap\Delta=\emptyset$,

\item ${T}\subseteq \Gamma^*$, $S\subseteq \Delta^*$, both are
poly-time decidable, and

\item $R\subseteq \Gamma^*\times \Delta^*$ is a binary relation.  %and
%there is $k$ such that, $R(t,w)$ implies that $|w|\leq |t|^k$ for
%any $t\in {\cal T}$ and any $w\in \Gamma^*$.
\end{itemize}

\noindent For a string $t\in {T}$, and a string $w\in S$, we say $w$
is a $R$-model of $t$ if $R(t,w)$ holds. The set of all $R$-models
of $t$ is denoted as Mod$_R(t)$

%If the unary relation $\exists w R(t,w)$ is in NP,\footnote{Please
%note that the assumption that $\exists w R(t,w)$ is in NP does not
%imply that $R(t,w)$ is in NP.} that is, if the problem of deciding
%whether a given $t\in {T}$ has a $R$-model is in NP, then we call
%$(\Gamma, {T}, R)$ a NP-representing system. We use ${\cal NP}$ to
%denote the class of all NP-representing systems
\end{definition}

Intuitively,  one may regard strings in $T$ as (encodings of) finite
theories (e.g., a propositional formula, or a logic program, etc.),
whereas strings in $S$ are intended to encode interpretations of
atoms. Then the predicate $R(t,w)$ says the interpretation encoded
by $w$ satisfies the theory encoded by $t$. That is, $R$ is a
satisfactory relation. Please note that from the satisfactory
relation $R$ we can define the following deductive relation: we say
$t_1$ entails $t_2$ if Mod$_R(t_1)\subseteq$ Mod$_R(t_2)$.

\vskip 2mm \noindent {\bf Example 1.} Let ${\cal L}:=\{x, |, \neg,
\wedge, \vee, \rightarrow, \exists,\forall, ), (\}$,
$\Delta=\{0,1\}$. We intend to use $x|, x||, x|||, \cdots$ to denote
the propositional variables $x_1, x_2, x_3, \cdots$. A string $t$ of
${\cal L}^*$ is an encoding of a propositional formula over
variables $x_1,\cdots,x_n$ if $t$ can be obtained from the formula
by replacing each occurrence of each $x_i$ by $x|\cdots|$ ($x$ is
followed by $i$ many $|$'s). Similarly, logic programs, quantified
Boolean formulas can be encoded in a natural way as strings in
${\cal L}^*$.

Further, $\{0,1\}$-sequences $w$ with length $n$ are intended to
code truth assignments $v$ on variables $x_1, x_2,\cdots, x_n$, more
precisely, $v(x_i)=1$ if and only if the $i$-th symbol in $w$ is 1.
Please note that a truth assignment is uniquely determined by a
subset of atoms and vice visa. Thus we can also consider a
$\{0,1\}$-sequence as an encoding of a  subset of atoms.

\begin{description}
\item (1) Let PF be the class of encodings of propositional
formulas, TA$:=\{0,1\}^*$, and Sat$(t,w)$ be the relation which says
that $w$ is an encoding of a satisfying truth assignment of the
formula coded by $t$. Then $({\cal L}$, $\{0,1\}$, PF, TA, Sat) is
in fact the propositional logic system.

\item (2) Let LP be the class of encodings of (propositional)
logic programs, ANS$(t,w)$ be the relation which says that $w$ is an
encoding of an answer set (see \cite{gelfond}) of the normal logic
program coded by $t$. Then ($\Gamma$, $\{0,1\}$, LP, TA, ANS) is the
answer set logic programming system.

\item (3) Let $\exists$PF be the class of encodings of formulas $\Phi$ of
the form $\exists x_1\cdots\exists x_m \varphi$ with $\varphi\in$ PF
(here free variables are allowed). And let FSat$(\Phi,w)$ be the
relation which says that $w$ is an encoding of a truth assignment
$v$ on free variables of $\Phi$, and after applying $v$ to $\Phi$
the resulting formula $\Phi[v]$ is true. Then $({\cal L}$,
$\{0,1\}$, $\exists$PF, TA, Fsat) is a logic system.
\end{description}
%end example

For simplicity, from now on we  write a logic system $(\Gamma,
\Delta, T, S, R)$ just as $(T, S, R)$. For example we will write
(PF, TA, Sat) instead of $({\cal L}, \{0,1\}$, PF, TA, Sat).

\begin{definition} \label{d-2} (\cite{zhkb}) \
Let $(T_1, S_1, R_1)$, $(T_2, S_2, R_2)$ be two logic systems. We
say $(T_1, S_1, R_1)$ can be {\bf poly-time model-equivalently
reduced to} $(T_2, S_2, R_2)$, denoted as $(T_1, S_1,
R_1)\preceq_{ptime} (T_2, S_2, R_2)$, if there are two polynomials
$p(n)$ and $q(n)$, a function $f : {T_1} \longrightarrow {T_2}$, and
a  mapping $g : T_1\times S_1 \longrightarrow S_2 $
%%%%
satisfying

\begin{itemize}
\item $f$ is computable in time $p(n)$, where $n$ is the size of
input theory $t$,
\item $g$ is computable in time $q(n)$, where $n$ is the size of
input $(t,w)$, and
\item for any fixed $t\in T_1$, the mapping $g_t$, defined by
$g_t(v):=g(t,v)$, is a bijection from Mod$_{R_1}(t)$ to
Mod$_{R_2}(f(t))$.
\end{itemize}

If in the above definition of $\preceq_{ptime}$ we replace ``$f$ is
computable in time $p(n)$" by ``$f$ is computable in space $p(n)$",
then the reduction is called poly-space model-equivalent reduction,
denoted as $\preceq_{pspace}$.

If the mapping $g$ in the definition of $\preceq_{ptime}$ satisfies
$g(t,v)=v$ for all $t\in {T}_1$ and $v\in S_1$, i.e., $t$ and $f(t)$
are equivalent, then the reduction is called poly-time equivalent
reduction, denoted as $\preceq^{equ}_{ptime}$. Likewise for
poly-space equivalent reduction $\preceq^{equ}_{pspace}$.

\end{definition}

Clearly, $\preceq_{ptime}$ and $\preceq_{pspace}$ are transitive.
And a poly-time model-equivalent reduction is also a poly-space
model-equivalent reduction.

\begin{lemma}\label{l-1} \ \
\begin{description}
\item (1) (PF, TA, Sat) $\preceq_{ptime}$ (CNF, TA, Sat).
\cite{plgr} %[Plaisted, Greenbaum]

\item (2) Suppose P$\not\subseteq$NC$^1$/poly, then
(LP, TA, ANS) $\not\preceq_{pspace}^{equ}$ (PF, TA,
Sat). \cite{lifraz}%[Robzev, Lifchitz]

\item (3) Suppose NP$\not\subseteq$P/poly, then ($\exists$PF, TA, FSat) $\not\preceq_{pspace}$ (PF, TA, Sat). \cite{zhkb}
\end{description}
\end{lemma}

\section{NP-Logic System}

For a logic system $(T, S, R)$, a theory $t\in T$ may have
$R$-models with super-polynomial size in the size of $t$. However,
in this paper we concentrate on systems such that any $R$-model of
every theory has polynomial size.

\begin{definition} \label{d-3} Suppose $(\Gamma, \Delta, T, S, R)$
is a logic system. If there is a polynomial $p$ such that $R(t,w)$
implies $|w|= p(|t|)$ for any $t\in {T}$ and any $w\in S$, then we
call $(\Gamma, \Delta, T, S, R)$ a poly-size system.
\end{definition}

Obviously, all systems in Example 1 are poly-size systems. From now
on, whenever speaking of a logic system we mean it is a poly-size
system.

\begin{definition}
A logic system $(T,S,R)$ is said to be a NP-logic system if the model-existence problem is in NP, i.e.,
the problem whether a given formula in $T$ has a $R$-model can be decided non-deterministically in polynomial time.
\end{definition}

Obviously, for a logic system $(T,S,R)$, if the model-checking problem is in NP (i.e. $R$ is in NP) then the system is an NP-system.
Therefore, all logic systems in Example 1 are NP-systems.

\begin{theorem}
\begin{description}
\item (1) For any logic system $(\Gamma, \Delta, T, S, R)$, if the
relation $R$ is decidable in polynomial time on a non-deterministic
Turing machine, then $(T, S, R) \preceq_{ptime}$ ($\exists$PF, TA,
FSat)

\item (2) For any logic system $(\Gamma, \Delta, T, S, R)$, if the
relation $R$ is decidable in polynomial time on a deterministic
Turing machine, then $(T, S, R) \preceq_{ptime}$ (PF, TA, Sat).
\end{description}
\end{theorem}

\begin{proof}
(1) At first we have to construct a transformation from $T$ to PF.
We shall adopt the construction in the proof of Cook-Levin theorem
(see e.g. \cite{sipser}) which states that the satisfiability
problem for propositional formulas is NP-complete. Let $R'\subseteq
\Gamma^* \times \Delta^*$ be the relation defined by $R'(t,w)$ if
and only if $t\in T$, $w\in S$ and $R(t,w)$. Since $T, S$ are both
decidable in polynomial time, $R'$ is still decidable
non-deterministically in
polynomial time. %%%%%%%%%
Let $N=(\Gamma\cup\Delta, Q, \Gamma', \sigma, q_0, q_{accept},
q_{reject})$ be a non-deterministic Turing machine that decides
$R'(t,w)$ in time $(|t|+|w|)^{k_0}$ for some constant $k_0$. Please
note that we have assumed that $(T,S,R)$ is poly-size system (see
Definition 3). Then there is a polynomial $p(n)$ such that $R(t,w)$
implies $|w|=p(|t|)$. Then on input $t,w$ with $|w|=p(|t|)$, the
configurations of a branch of the computation can be represented as
an $((|t|+p(|t|))^{k_0}+1)\times ((|t|+p(|t|))^{k_0}+1)$ table. As
shown in the following figure, the first row of the table is the
starting configuration of $N$ on input $w,t$, and each row follows
the previous one according to the transition function $\sigma$.

\begin{center}
\includegraphics[width=5in,height=2.3in]{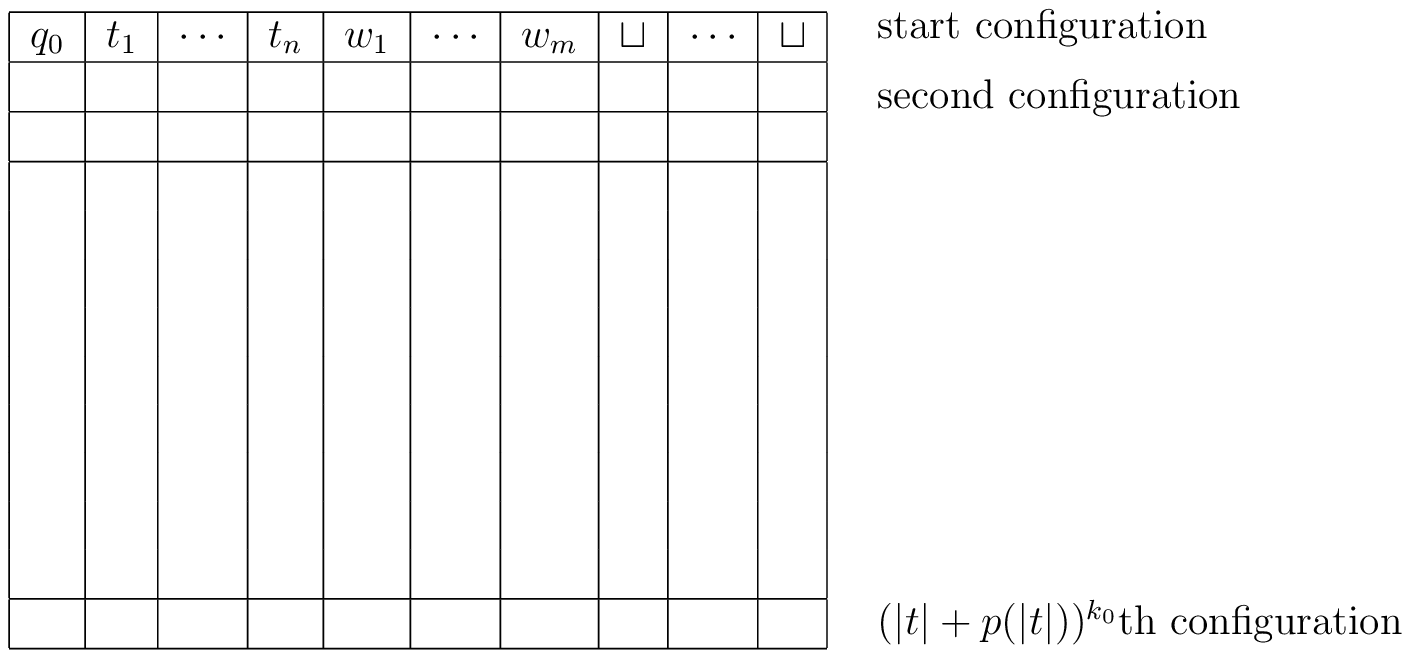}%
\remove{
\begin{tabular}{|c|c|c|c|c|c|c|c|c|c|}\hline
$q_0$ & $t_1$& $\cdots$& $t_n$ & $w_1$ &$\cdots$&$w_m$&$\sqcup$&$\cdots$&$\sqcup$\\
\hline &&&&&&&&&
\\ \hline &&&&&&&&&\\ \hline &&&&&&&&&\\ &&&&&&&&&\\ &&&&&&&&&\\ &&&&&&&&&\\ &&&&&&&&&\\&&&&&&&&&
\\ &&&&&&&&& \\  \hline &&&&&&&&&\\ \hline
\end{tabular}
}
\end{center}

For any $i, j$ with $1\leq i,j\leq (|t|+p(|t|))^{k_0}+1$ and for
each symbol $s\in \Gamma'\cup Q$, we have a propositional variable
$x_{i,j,s}$. If $x_{i,j,s}$ take the value 1, it means that the
entry (or cell) in row $i$ and column $j$ contains the symbol $s$.

In the proof of Cook-Levin theorem, for input $t,w$, four
propositional formulas $\varphi_{{cell}}$, $\varphi_{{start}}$,
$\varphi_{{move}}$, and $\varphi_{{accept}}$ are designed so that
the Turing machine accepts $w, t$ if and only if
$\varphi_{{cell}}\wedge \varphi_{{start}}\wedge
\varphi_{{move}}\wedge\varphi_{{accept}}$ is satisfiable.

$\varphi_{{start}}$ states that the first row of the table is the
starting configuration of $N$ on input $t=t_1t_2\cdots
t_n,w=w_1w_2\cdots w_{p(n)}$. More precisely,

$$\begin{array}{lll}\varphi_{{start}}&:=& x_{1,1,q_0}\wedge x_{1,2,t_1}\wedge\cdots\wedge x_{1,n+1,t_n}\wedge\\
&&x_{1,n+2,w_1}\wedge\cdots\wedge x_{1,n+p(n)+1,w_{p(n)}}\wedge\\ &&
x_{1,n+p(n)+2,\sqcup}\wedge\cdots\wedge x_{1, n',
\sqcup}.\end{array}$$ Here $n'$ is $(n+p(n))^{k_0}+1$.

We need not to write explicitly other three formulas, instead we
just explain their intuitive meaning. $\varphi_{{cell}}$ states that
each cell contains exactly one symbol. The formula
$\varphi_{{move}}$ guarantees that each row of the table corresponds
to a configuration that can be obtained from the preceding row's
configuration by applying a rule of $N$. Finally,
$\varphi_{{accept}}$ states that an accepting configuration occurs
during the computation.

Please note that the above construction of $\varphi_{{start}}$
depends on the input information $t,w$. However, our task is to
construct a mapping which transforms each $t\in T$ to an
existentially quantified formulas. That is, our construction should
not depend on $w$. For that reason, we have to modify the
formula $\varphi_{{start}}$. Please note the $w$ could be any string
in $\Gamma^*$ with length $p(|t|)$. Hence, each cell in the first
row and the $(n+1+j)$-th (with $1\leq j\leq p(n)$) column could
contain
any symbol of $\Gamma\cup\Delta$. %, but whenever a cell contains the
%empty symbol $\sqcup$, all cells behind it must contain the symbol
%$\sqcup$.
This can be described as the following formula:
$$\alpha:=\left(\bigwedge_{n+2\leq i\leq n+p(n)+1}\left(\bigvee_{s\in
\Gamma\cup\Delta}x_{1,i,s}\right)\right)%\wedge
%
%\left(\bigwedge_{n+2\leq i\leq k'}\left(x_{1,i,\sqcup}\rightarrow
%\bigwedge_{i\leq j\leq k'} x_{1,j,\sqcup}\right)\right)
$$
Now we define $\varphi'_{{start}}$ as
$$\varphi'_{{start}}:= x_{1,1,q_0}\wedge x_{1,2,t_1}\wedge\cdots\wedge x_{1,n+1,t_n}\wedge \alpha
\wedge x_{1,n+p(n)+2,\sqcup}\wedge\cdots\wedge x_{1, n', \sqcup}.$$

We write $G(t)$ to denote the formula $\varphi_{{cell}}\wedge
\varphi'_{{start}}\wedge \varphi_{{move}}\wedge\varphi_{{accept}}$.
Clearly, $t$ has a $R$-model if and only if $G(t)$ is satisfiable.
Please note that for a string $w$, the truth values of $x_{1,j,s}$
(with $1\leq j\leq n'$) can be uniquely determined by $t,w$ and
formulas $\varphi_{{cell}}$ and $\varphi'_{{start}}$, however, the
truth value of each of the other variables is not uniquely determined due to
the non-determinism of $N$.  Thus, the models of $t$ do not
necessarily one-to-one correspond to the truth assignments of
$G(t)$. Therefore, we add existential quantifiers $\exists
x_{i,j,s}$ in front of $G(t)$ for all $i=2,\cdots, n'$, $j=1,\cdots,
n'$ and $s\in \Gamma\cup Q$, the resulting formula is denoted as
$F(t)$. Now, it is easy to see that there is a polynomial-time
computable one-to-one correspondence between $R$-models of $t$ and
models of $F(t)$.

(2) Suppose $R$ is decidable deterministically in polynomial time.
Then, in the above construction, we can assume that $N$ is a
deterministic Turing machine deciding $R'$. Since the computation of
$N$ is uniquely determined whenever the input is fixed, we can see
that $G$ is in fact a model-equivalence reduction from $(\Gamma,
\Delta, T, S, R)$ to (${\cal L}$, $\{0,1\}$, PF, TA, Sat).
\end{proof}

%%%%%%%%%%%%%%%%%%%%%%%%%%%%%%%%%%%%%%%%%%
\vskip 2mm Theorem 1 says that $\exists$PF is the maximal system
amongst NP-systems for which the model checking problem is in NP.
However, it is unlikely a maximal NP-system in general because there
are NP-systems for which the model checking problem is
co-NP-complete or even harder (under the assumption that the
polynomial hierarchy does not collapse).

\vskip 2mm \noindent {\bf Example 2.}
\begin{description}
\item (1) Let MinSat$(t, w)$ be the relation which says that $w$
is an encoding of a minimal model of the propositional formula with
code $t$. Here by a minimal model of a propositional formula
$\varphi$, we mean a model $M$ of $\varphi$ such that any proper
subset of $M$ is not a model of $\varphi$. Since a propositional
formula $\varphi$ has a model if and only if $\varphi$ has a minimal
model, it follows that (${\cal L}$, $\{0,1\}$, PF, TA, MinSat) is a
NP-system.

\item (2) (${\cal L}$, $\{0,1\}$, $\exists$PF, FMinSat) is a
NP-system. Here FMinsat$(t,w)$ says $w$ is an encoding of a minimal
model of the existentially quantified formula coded by $t$.
\end{description}

\begin{proposition}
\begin{description}
\item (1) The model checking problem for (PF, TA,
MiniSat) is co-NP-complete \cite{cds}.%[Cadoli].

\item (2) The model checking problem for ($\exists$PF, TA, FMiniSat) is
$D^P$-complete. Where $D^P$ is the class of decision problems which
can be described as the intersection of one NP problem and one co-NP
problem \cite{papaw}.
\end{description}
\end{proposition}
\begin{proof} (1) Please see page 48-49 in \cite{cds}.%[Cadoli].

(2) At first we show the membership. Consider an arbitrary
$\exists$PF formula $\Phi$. Suppose the set of free (i.e. not
quantified) variables in $\Phi$ is $Z=:\{z_1,\cdots,z_n\}$, and the
set of bounded variables in $\Phi$ is $X:=\{x_1,\cdots,x_m\}$. For
simplicity we write $\Phi$ as $\exists X\varphi(X,Z)$. It is not
hard to see that a subset $M\subseteq Z$ of is a minimal model of
$\Phi$ if and only if $M$ is a model of the following formula.
$$\exists X\varphi(Z,X) \wedge \forall Z'(\neg (Z'\rightarrow Z)\vee
\forall X\neg \varphi(Z',X)\vee (Z'=Z)).$$
Where $Z'=\{z'\mid z\in Z\}$ is a set of new variables;
$Z'\rightarrow Z$ is abbreviated for the formula $\bigwedge_{z\in
Z}(z'\rightarrow z)$; $Z'=Z$ denotes the formula $(Z'\rightarrow
Z)\wedge (Z\rightarrow Z')$; and $\varphi(Z',X)$ is the formula
obtained from $\varphi(Z,X)$ by replacing each occurrence of
$z$ by $z'$. It follows obviously that minimal model checking
problem for $\exists$PF is in $D^P$.

Next we show the hardness. The canonical $D^P$-complete problem is
the SAT-UNSAT problem (see \cite{papaw}) of determining for a pair
$(\varphi, \psi)$ of propositional formulas, whether $\varphi$ is
satisfiable and $\psi$ is unsatisfiable. Let $X,Y$ be the sets of
variables in $\varphi$ and $\psi$, respectively. We assume w.o.l.g.
that $X\cap Y=\emptyset$. Let $z$ be a new variable. Consider the
following formula
$$F=\exists X\exists Y(\varphi\wedge(z\vee\psi)).$$
It is not hard to see that $(\varphi,\psi)\in$ SAT-UNSAT if and only if
$\{z\}$ is a minimal model of $F$. The proof completes.
\end{proof}

\begin{theorem}
Suppose co-NP $\not=$ NP. Then (PF, TA, MinSat) and $(\exists$PF,
TA, FSat) are pairwise incomparable with respect to poly-time
model-equivalence reduction.
\end{theorem}
\begin{proof}
The theorem follows from the following fact. The the minimal model
checking problem %of determining whether a model is a minimal model
%of
for propositional formulas is co-NP-complete (see Proposition 1),
whereas the model checking problem for existentially quantified
propositional formulas is NP-complete \cite{zhkb}. Suppose for
example (PF, TA, MinSat) $\preceq_{ptime}$ $(\exists$PF, TA, FSat).
Then for a truth assignment $M$ and a propositional formula
$\varphi$, to check that $M$ is a minimal model of $\varphi$, we
first transform $\varphi$ in poly-time into a $\exists$PF-formula
$\Phi$, and compute $M'$ from $M$ by using the poly-time computable
one-to-one correspondence, then check that $M'$ is a model of
$\Phi$, which is a NP problem. It follows that NP=co-NP, contradicts
the assumption of the theorem.
\end{proof}

\vskip 2mm NP$\not\subseteq$P/poly is an important conjecture in
computational complexity theory (see e.g. \cite{papa}). In fact we
even do not know whether NP$\not\subseteq$co-NP/poly is true or
false. However, the following theorem shows that if (PF, TA, MinSat)
and $(\exists$PF, TA, FSat) are comparable with respect to
poly-space model-equivalent reduction then NP$\subseteq$co-NP/poly.

\begin{theorem}
Suppose NP$\not\subseteq $co-NP/poly. Then (PF, TA, MinSat) and
$(\exists$PF, TA, FSat) are pairwise incomparable with respect to
poly-space model-equivalent reduction.
\end{theorem}
\begin{proof}
We first show $(\exists$PF, TA, FSat)$\not\preceq_{pspace}$(PF, TA,
MinSat).
Let $\Gamma_n$ be the set of all 3CNF formulas $\varphi$ such that
variables $|\varphi|=n$ and $var(\varphi)\subseteq\{x_1,\cdots,x_n\}$, where
$var(\varphi)$ is the set of all variables occurring in $\varphi$.
Define $\Gamma:=\bigcup_{n>0}\Gamma_n$. Clearly, the satisfiability
problem for $\Gamma$ is NP-complete.

Let $\pi(n)$ be the set of 3-clauses over $x_1,\cdots,x_n$. For each
3-clause $c\in \pi(n)$ introduce a new variable $z_c$. Define
$$\Psi_n:=\exists x_1\cdots\exists x_n \left(\bigwedge_{c\in\pi(n)} (c\vee
\neg z_c)\right)$$

Let $\varphi$ be a 3CNF formula with $|\varphi|=n$. W.l.o.g. we can
assume $\varphi\in \Gamma_n$. Suppose $\varphi=c_1\wedge\cdots\wedge
c_k$. Define $M_{\varphi}=\{z_{c_1},\cdots,z_{c_k}\}$, that is, we
set each $z_{c_i}$ to 1, and all other $z_c$ to 0. Clearly,
\begin{itemize}
\item $M_{\varphi}$ can be computed in polynomial time, and
\item $\varphi$ is satisfiable if and only if $M_{\varphi}$ is a
model of $\Psi_n$.
\end{itemize}

Suppose $(\exists$PF, TA, FSat)$\preceq_{pspace}$(PF, TA, MinSat).
Then there is a sequence $\psi_1,\psi_2,\cdots, \psi_n,\cdots$ of
propositional formulas such that
\begin{itemize}
\item  the size of each $\psi_n$ is bounded by a polynomial, and

\item for each $n$, there is a polynomial-time computable
one-to-one correspondence between the models of $\Psi_n$ and
$\psi_n$.
\end{itemize}
Then we define an advice-taking Turing machine \footnote{For a
precise definition of advise-taking Turing machine please see
\cite{papa}.} in the following way. The advice oracle is $\psi_n$.
Given an instance $\varphi$ of $\Gamma$ with $|\varphi|=n$, the machine
loads $\psi_n$, then computes $M_{\varphi}$ in polynomial time in
$n$, then computes $M'_{\varphi}$ according to the one-to-one
correspondence, finally checks whether $M'_{\varphi}$ is a minimal
model of $\psi_n$. Please note that the minimal model checking
problem for propositional formulas is in co-NP. Since the
satisfiability problem for $\Gamma$ is NP-complete, it follows that
NP$\subseteq$co-NP/poly.

(2) Next we show (PF, TA, MinSat) $\not\preceq_{pspace}$
($\exists$PF, TA, FSat). Suppose $\Gamma_n,\ \Gamma,\ \pi(n),\ z_c$
are defined as before. Now for each $c\in\pi(n)$ we introduce
another new variable $z'_c$ for each $c\in\pi(n)$, and a new
variable $y$ in addition. Define
$$\psi_n:=\left(\left(\neg y\wedge\bigwedge_{c\in\pi(n)}(c\vee\neg z_c)\right)\vee
(y\wedge x_1\wedge\cdots\wedge x_n) \right)\wedge
\left(\bigwedge_{c\in\pi(n)}(z_c\leftrightarrow \neg
z'_c)\right).$$

Let $\varphi$ be a 3CNF formula with $|\varphi|=n$. W.l.o.g. we
can assume $\varphi\in \Gamma_n$. Suppose
$\varphi=c_1\wedge\cdots\wedge c_k$. Define
$$M_{\varphi}=\{y,x_1,\cdots,x_n\}\cup\{z_{c_1},\cdots,z_{c_k}\}\cup\{z'_c\mid
c\not\in\varphi\}.$$%, that is, we set $z_{c_i}$ to 1, and all other
%$z_c$ to 0. Clearly,
It is not hard to see that
\begin{itemize}
\item $M_{\varphi}$ can be computed in polynomial time, and

\item $\varphi$ is unsatisfiable if and only if $M_{\varphi}$ is a
minimal model of $\Psi_n$.
\end{itemize}

Suppose (PF, TA, MinSat) $\preceq_{pspace}$ $(\exists$PF, TA, FSat).
Then there is a sequence $\Psi_1,\Psi_2,\cdots, \Psi_n,\cdots$ of
$\exists$PF formulas such that
\begin{itemize}
\item  the size of each $\Psi_n$ is bounded by a polynomial, and

\item for each $n$, there is a polynomial-time computable
one-to-one correspondence between the models of $\psi_n$ and
$\Psi_n$.
\end{itemize}
Then we define an advice-taking Turing machine in the following way.
The advice oracle is $\Psi_n$. Given an instance $\varphi$ of $\Gamma$
with $|\varphi|=n$, the machine loads $\Psi_n$, then computes
$M_{\varphi}$ in polynomial time in $n$, then computes
$M'_{\varphi}$ from $M_{\varphi}$ according to the one-to-one
correspondence, finally checks whether $M'_{\varphi}$ is a model of
$\Psi_n$. Please note that the model checking problem for
$\exists$PF formulas is in NP. Since the unsatisfiability problem
for $\Gamma$ is co-NP-complete, it follows that NP$\subseteq$co-NP/poly.
\end{proof}

\begin{lemma}
\begin{description}
\item (1) (PF, TA, MinSat) $\preceq_{ptime}$ ($\exists$PF, TA,
FMinSat).

\item (2) ($\exists$PF, TA, FSat) $\preceq_{ptime}$ ($\exists$PF, TA,
FMinSat).
\end{description}
\end{lemma}
\begin{proof}
(1) Directly follows from the fact that (PF, TA, MinSat) is a
sub-system of ($\exists$PF, TA, FMinSat).

(2) Consider any formula $\Phi=\exists y_1\cdots\exists y_m \varphi$
with free variables $x_1,\cdots, x_n$. Now we introduce for each
$x_i$ ($i=1,\cdots,n$) a new variable $x'_i$ which is intended to
stand for $\neg x_i$. Define $\varphi':=\varphi\wedge \bigwedge
\left( (x_i\vee x'_i)\wedge (\neg x_i\vee \neg x'_i)    \right)$ and
let $\Phi':=\exists y_1\cdots\exists y_n\varphi'$. Clearly a subset
$M\subseteq \{x_1,\cdots, x_n\}$ is a model of $\Phi$ if and only if
$M\cup\{x'_i\mid x_i\not\in M\}$ is a minimal model of $\Phi'$.
\end{proof}

\begin{corollary} Suppose co-NP $\not=$ NP. Then
\begin{description}
\item (1) ($\exists$PF, TA, FMiniSat) $\not\preceq_{ptime}$ (PF, TA,
MinSat).

\item (2) ($\exists$PF, TA, FMinSat) $\not\preceq_{ptime}$ ($\exists$PF,
TA, FSat).
\end{description}
\end{corollary}
\begin{proof} If ($\exists$PF, TA, FMiniSat) $\preceq_{ptime}$ (PF, TA, MinSat)
or ($\exists$PF, TA, FMinSat) $\preceq_{ptime}$ ($\exists$PF, TA,
FSat), then we have by Lemma 2 that (PF, TA, MinSat) and
($\exists$PF, TA, FSat) are comparable w.r.t. poly-time
model-equivalent reduction. This contradicts Theorem 2.
\end{proof}

\begin{corollary}
Suppose NP$\not\subseteq $co-NP/poly. Then
\begin{description}
\item (1) ($\exists$PF, TA, FMiniSat) $\not\preceq_{pspace}$ (PF, TA,
MinSat).

\item (2)($\exists$PF, TA, FMinSat) $\not\preceq_{pspace}$ ($\exists$PF,
TA, FSat). \end{description}%
\end{corollary}
\begin{proof}
If ($\exists$PF, TA, FMiniSat) $\preceq_{pspace}$ (PF, TA, MinSat)
or ($\exists$PF, TA, FMinSat) $\preceq_{pspace}$ ($\exists$PF, TA,
FSat), then we have by Lemma 2 that (PF, TA, MinSat) and
($\exists$PF, TA, FSat) are comparable w.r.t. poly-space
model-equivalent reduction. This contradicts Theorem 3.
\end{proof}

\section{Conclusion and Future work}

We have proved that w.r.t. poly-time model-equivalent reduction
($\exists$PF, TA, FSat) has the strongest expressive power among NP
systems with model checking problem in NP, whereas (PF, TA, Sat) is
strongest among NP systems with model checking problem in P.
However, ($\exists$PF, TA, FSat) is unlikely the strongest NP
system, because it have been shown that ($\exists$PF, TA, FSat) $\preceq_{ptime}
(\exists$PF, TA, MinSat) but the converse in not true under the
assumption NP$\not\subseteq$ co-NP. We conjecture that there is no
strongest NP-system under some conjecture in computational
complexity.

{}
\end{document}